\documentclass[11pt]{article} 
\usepackage[utf8]{inputenc} 

\usepackage{geometry} 
\usepackage{caption}
\usepackage{subcaption}
\usepackage{float}
\usepackage{balance}
\usepackage{comment}
\usepackage{booktabs}
\usepackage{times}
\usepackage{hyperref}
\usepackage{url}

\usepackage{fullpage}

\usepackage{amsmath,amsthm,amssymb}
\usepackage{algorithm}
\usepackage{algpseudocode}
\usepackage{amsmath,amssymb,array,color,colortbl,graphicx,multirow}
\usepackage{paralist}

\newtheorem{lemma}{Lemma}

\newtheorem{corollary}{Corollary}
\theoremstyle{remark}

\usepackage{xspace,color,colortbl}

\title{The Evolutionary Price of Anarchy:
\\ Locally Bounded Agents in a Dynamic Virus Game}

\author{Krishnendu Chatterjee$^1$ \quad Laura Schmid$^1$ \quad Stefan Schmid$^2$\\
{\small $^1$ IST Austria \quad $^2$ Faculty of Computer Science, University of Vienna}  
}
\date{}

\begin{document}

\thispagestyle{empty}

\maketitle

\thispagestyle{empty}

\begin{abstract}
The Price of Anarchy (PoA) is a well-established 
game-theoretic concept to shed light on 
coordination issues arising in open distributed systems. 
Leaving agents to selfishly optimize 
comes with the risk of ending up in 
sub-optimal states (in terms of performance and/or costs), 
compared to a centralized system design. 
However, the PoA relies on strong assumptions 
about agents' rationality (e.g., resources and information) and
interactions, whereas in many distributed systems
agents interact locally with bounded resources. They do so repeatedly over time (in contrast to
``one-shot games''), and their strategies may evolve.

Using a more realistic evolutionary game
model, this paper introduces a realized 
\emph{evolutionary Price of Anarchy (ePoA)}.
The ePoA allows an exploration of equilibrium selection in dynamic distributed
systems with multiple equilibria, based on \emph{local} 
interactions of simple \emph{memoryless} agents.

Considering a fundamental game related to virus propagation
on networks, we present analytical bounds on the ePoA in 
basic network topologies and for different strategy update dynamics. 
In particular, deriving stationary distributions of the stochastic evolutionary process, 
we find that the Nash equilibria are not always the 
most abundant states, and that different processes
can feature significant off-equilibrium behavior, leading to a significantly
higher ePoA compared to the PoA studied traditionally in the literature.
\end{abstract}

\section{Introduction}

The performance and efficiency of large distributed systems, 
such as open peer-to-peer networks which any user can join, 
often critically depend on cooperation and can suffer if users
behave selfishly, e.g.:
consume but not contribute 
resources~\cite{bitthief}, 
choose routes~\cite{roughgarden2002bad} and neighbors~\cite{fabrikant2003network} 
to optimize their \emph{personal} benefits, etc.
Non-cooperative behavior can also affect security. For example, 
if security mechanisms such as the installation of anti-virus software 
are employed just for self-protection, a virus may spread further than in
cooperative environments~\cite{fleizach2007can}, while at the same
time increasing global security investment costs~\cite{Aspnes:2005:ISV:1070432.1070440}. 

The \emph{Price of Anarchy (PoA)}~\cite{koutsoupias1999worst} is a game-theoretic concept which
allows to assess to which extent a distributed
system is affected negatively by non-cooperative behavior.
Essentially, the PoA compares the optimal social welfare resp.~cost
to the welfare resp.~cost in the worst
Nash equilibrium: an equilibrium in which no selfish agent, henceforth
called \emph{player}, has an incentive
to change its behavior. If a distributed system
has a large PoA, this means that the performance resp.~cost
can be far from optimal: the system may require a redesign
or at least strong incentive mechanisms. 

However, Nash equilibria are often not a good model for
real-world distributed systems, for several reasons. In particular:
\begin{compactenum}
    \item \textbf{\emph{Dynamic character:}} Distributed systems typically 
		are not based on ``one-shot games'' but rely on dynamic interactions over time:
e.g., peers (i.e., \emph{players}) in peer-to-peer systems 
such as BitTorrent interact repeatedly, for example
using tit-for-tat strategies, leading to \emph{repeated games}~\cite{weibull1997evolutionary}.

\item \textbf{\emph{Local information:}} Players in 
distributed systems often only have local information
about the network, based on interactions (e.g., with neighboring players).
Obtaining global information
is typically infeasible, especially in large-scale distributed systems.

\item \textbf{\emph{Bounded resources:}} Players typically also have only limited resources,
e.g., in terms of memory or in terms of the complexity of the kind of
algorithms they can execute.  
\end{compactenum}

This paper is motivated by the desire to extend the 
concept of price of anarchy
to account for these important characteristics of distributed systems.
In particular, we aim to port the PoA to
\emph{evolutionary games} and \emph{local information} scenarios:
games with simple, \emph{memoryless} players which interact repeatedly and locally, and can update their strategies over time.  
This is timely, and the research community is currently discussing 
alternatives to Nash equilibria such as Bayes-Nash equilibria~\cite{singh2004computing} 
for games with incomplete information; 
at the same time, it is believed
that such extensions are complex due to having to take into account players' belief systems, amongst other issues, and thus introduce major research challenges~\cite{roughgarden2015price}. 

In fact, analyzing equilibrium selection in stochastic processes as described by evolutionary
games is already challenging in 2-player games in a population with $m$ pure strategies~\cite{fudenberg2006evolutionary,FUDENBERG2008229}. 
Games 
\emph{on graphs}, while also highly useful in verification and synthesis of (possibly distributed~\cite{pneuli1990distributed}) reactive systems~\cite{chatterjee2012survey,mohalikdist}, are often particularly difficult, also when considering evolutionary games~\cite{allen2014games},
due to the additional dependencies on the possible interactions.

\subsection{Our contributions}

This paper extends 
the notion of price of anarchy 
to evolutionary games, introducing the 
 \emph{evolutionary Price of Anarchy (ePoA)}.
In particular, we are interested in the equilibrium
behavior of simple memoryless players, 
which repeatedly and \emph{locally} 
interact on a graph. 
The ePoA is essentially a framework and can be used
to study different evolutionary dynamics 
(the parameters to our framework) and ways in which
players evolve their strategies.

To shed light on
how the evolutionary perspective can affect the
conclusion on the possible impact of selfish behavior
in distributed systems, we consider a concrete
case study: the well-known virus propagation game
introduced by Aspnes et al.~\cite{Aspnes:2005:ISV:1070432.1070440}.
We present an analysis of the evolutionary dynamics
of this game 
for the three fundamental dynamic models (related to genetic evolution
and imitation dynamics) and different basic network topologies
(cliques, stars, and beyond). 
Interestingly,  while the analysis of such evolutionary games
is challenging in general, we are able to provide
an exact characterization of the long-run frequencies of configurations for these scenarios. 

We make several interesting observations.
We find that the evolutionary dynamics of this game
give rise to a rich variety 
of behaviors. 
In particular, the ePoA
can be significantly worse than the classic PoA, 
for reasonable (i.e. not too high) mutation rates.
We also find that   
Nash equilibria are not always the 
most frequent, i.e., \emph{abundant}, states, and different processes
can feature significant off-equilibrium behavior.

Our analytical results are complemented with simulations, also of more complicated topologies.

\subsection{Organization}

The remainder of this paper is organized as follows.
Section~\ref{sec:prelim} introduces preliminaries.
We present the concept of evolutionary price of anarchy in
Section~\ref{sec:evol-poa} and study its application 
in the virus inoculation game in Section~\ref{sec:virusgame}.
Section~\ref{sec:appendsimresults} reports on simulation results
on more complex topologies.
After reviewing related work in Section~\ref{sec:relwork},
we conclude our work in Section~\ref{sec:conclusion}.
To improve readability, some technical details only appear
in the Appendix.
More detailed numerical results 
and our implementation are available upon request. 

\section{Preliminaries}\label{sec:prelim}

Before launching into a full description and analysis of our 
model, we  first revisit the virus inoculation game which will
serve us as a case study in this paper.
We will also
give a short introduction to evolutionary 
dynamics, evolutionary games and evolutionary graph theory. 

\subsection{The virus inoculation game}

In the classic one-shot virus inoculation game~\cite{Aspnes:2005:ISV:1070432.1070440},
nodes must choose between 
installing anti-virus software (inoculating themselves) at a cost, 
or risk infection by a virus that spreads from a random location and can reach a node 
if there is a path of not inoculated nodes in between. 
The network is modeled by an undirected graph $G=(V,E)$ with $N$ nodes. 
Nodes correspond to players in the game. Every node is equipped 
with a strategy $a_i$ that denotes its propensity to inoculate itself. 
If $a_i$ is 0 or 1, the strategy is called \emph{pure}. 
Every node takes an action according to its strategy; 
the overall configuration of nodes is reflected in 
the strategy profile $\Vec{a} \in [0,1]^n$. Inoculation 
costs $V$. After everyone 
has made their choice, the adversary picks a random node 
as the starting point of infection. The virus propagates 
through the network, infecting all non-inoculated nodes 
that have a direct path to an infected node. That is, propagation happens 
on an "attack graph" $G_{\vec{a}}$, where inoculated nodes 
have been removed, and only insecure nodes remain. Being 
infected then comes with a cost $I>V$. Inoculation prevents 
infection as well as further virus transmission by the inoculated node. 

The cost of a mixed strategy for a node $i$ in this model is therefore 
\begin{equation}
  \text{cost}_i(\Vec{a})=a_i V + (1-a_i)I \cdot p_i(\Vec{a}),  
\end{equation}

\noindent where $p_i(\Vec{a})$ is the probability of node $i$ 
being infected given the strategy profile $\Vec{a}$ and the 
condition that $i$ did not inoculate. The goal of each player 
is to minimize its own cost, while it does not take the 
resulting social cost to the system in total into account. This social cost is simply

\begin{equation}
    \text{cost}(\Vec{a}) = \sum_{j=0}^{N-1}\text{cost}_j(\Vec{a}).
\end{equation}

Aspnes et al.~then showed the following characterization of pure equilibria (for the proof and the extension to mixed equilibria, cf.~\cite{Aspnes:2005:ISV:1070432.1070440}):

\begin{corollary}[Characterization of pure equilibria]
Fix $V$ and $I$, and let the threshold be $t = VN/I$. A strategy 
profile $\Vec{a}$ is a pure Nash equilibrium if and only if:
\begin{compactenum}
    \item[(a)] Every component in the attack graph
    $G_{\Vec{a}}$ has at most size $t$.
    \item[(b)] Inserting any secure node $j$ and its 
    edges into $G_{\Vec{a}}$ yields a component of size at least $t$.
\end{compactenum}
\end{corollary}

\subsection{Evolutionary dynamics and games}

Game theory considers individuals that consciously aim to reach the best outcome for them in a strategic decision situation. Its classic framework, using concepts like Nash equilibria 
usually makes some key assumptions about the players' rationality, their beliefs and cognitive abilities. In contrast, 
\emph{evolutionary game theory}, as a generic approach to evolutionary dynamics~\cite{smith1973logic,smith1988evolution,hofbauer1998evolutionary}, considers a population of players with bounded rationality instead. Each player adopts a strategy to interact with other population members in a game. The players' payoffs from these interactions -- which depend on the actions of the co-players and therefore on the abundance of different strategies -- are considered to be their 
\emph{evolutionary fitness}. Success in the game is translated into reproductive success: good strategies reproduce faster, whereas disadvantageous strategies go extinct. In a nutshell, evolutionary game theory describes dynamics that are dependent on the frequency of different strategies in a population. 

The evolutionary dynamics depends on the setup and structure of the population, the underlying game, and the way strategies spread. There is a distinction to be made between deterministic dynamics in the limit of infinite populations~\cite{hofbauer1990adaptive}, described by differential equations (\cite{hofbauer1998evolutionary}), and stochastic dynamics in the case of finite populations~\cite{imhof2010stochastic} that take into account the intrinsic randomness of evolutionary processes. 
Under appropriate conditions, stochastic models however approach their deterministic counterparts as the population size becomes large~\cite{traulsen2009stochastic}.

Furthermore, the way strategies spread crucially depends on the way members of the population are connected to each other. \emph{Evolutionary graph theory} considers different population structures and how the dynamics are changed by changing the underlying graph~\cite{lieberman2005evolutionary}. Vertices now represent the members of the population, and edges the spatial/social connections between them. 
This is a natural model for multi-player systems, as interactions then happen only between neighbors. The stochastic process is described by a Markov chain on the states of the graph. 
However, analytical results on these properties in games on arbitrary graphs are difficult to obtain in general; many of the fundamental underlying problems are computationally hard (e.g.,
\cite{Ibsen-Jensen15636}).

\section{The Evolutionary Price of Anarchy}\label{sec:evol-poa}

It may not be realistic to assume 
that nodes in a game such as the virus game in~\cite{Aspnes:2005:ISV:1070432.1070440} will have perfect information in the 
first place. In large networks, it is unlikely that
nodes separated by many neighbors would know $G$ and 
each others' decision and would then optimally react to 
this large amount of information. Rather, 
it is more natural to think that nodes only 
have \emph{local} information. They 
can see their neighbors' most recent choices and react to them only when updating their strategy, 
while being unaware of the choices of nodes with a higher degree of separation. To model this, we first introduce an evolutionary virus inoculation game and three different kinds of stochastic evolutionary dynamics. We then define a general notion of the evolutionary price of anarchy.

\subsection{The evolutionary virus inoculation game}

We consider an evolutionary process on a static graph $G=(V,E)$ with $V$ the set of vertices (players) and $E$ the set of edges. The $N=|V|$ vertices are occupied by players in the game, and we say that two nodes/players are \emph{neighbors} if there is an edge connecting them. One iteration of the evolutionary process consists of three stages:

\begin{compactenum}

\item \emph{Decision making}. All players make a decision whether to inoculate themselves against possible virus infections in the network. In case they choose to do so, 
they pay a cost $V$, and pay nothing otherwise. Players' propensity to inoculate is encoded in their strategy~$a_i$.

\item \emph{Virus propagation.} After everyone has concurrently made a decision, expected costs of the nodes when the system is attacked by a virus are calculated. To do so, we use a process with $n$ steps: in each step, the virus starts at a different node of the graph and spreads throughout the network. Inoculated players are unaffected by the virus and cannot transmit it to their neighbors, while unprotected players pay a cost once they become infected, and spread the virus to their other insecure neighbors. Uninfected players who are not inoculated do not pay anything. We will use the term "realized cost vector" to describe the vector $ct=[I,V,L=0]$, where infected nodes pay $I$, inoculated nodes pay $V$, and insecure but uninfected nodes pay nothing. 
Once the virus has swept over the system, infecting unprotected players and their unprotected neighbors, costs are recorded, infection status is reset, 
and the next virus starts at the next node, until every node has served as a starting point. Alternatively, the process can be randomized by letting sufficiently many viruses start at random nodes. Once this has happened, cumulative costs are averaged, giving expected negative payoffs for the current strategy configuration, and the next stage commences. 

\item \emph{Evolution of strategies.} After receiving their expected negative payoff, players can assess the damage done and subsequently change their strategy based on comparisons with their neighbors' payoffs, before they again decide which action to take, and the next game begins. It is important to realize that this updating process is based on purely local information: nodes only need to know their neighbors' payoffs in order to make their decisions. This means that also the outcomes that can be realized may differ from the Nash equilibria that are found in the perfect information model of Aspnes et al. We consider dynamics that can take both \emph{selection} and \emph{mutation} into account: strategies typically evolve according to their (frequency-dependent) fitness, but they can also randomly change with a small probability $\mu$, which we will refer to as the \emph{mutation rate}. This prevents absorbing states in the process and lets us compute a \emph{unique} stationary distribution that gives the long-term frequencies of system configurations. In the limiting case of $\mu \rightarrow 0$, the process always terminates and reaches a state where either all nodes are inoculated, or none are. 
\end{compactenum}

We differentiate here between two kinds of well known memoryless evolutionary dynamics, but in general we can configure our framework for general monotone imitation dynamics as described in~\cite{FUDENBERG2008229}: 

\begin{algorithm}[t]
\caption{\sf Moran Death-Birth process  }
\label{alg:morandb}
\begin{algorithmic}[1]
\State \textbf{Global Parameters:} $N$, $k$ \Comment{Number of players, maximum updates in one generation}
\State{time $\gets 0$}
\While{time < $k$}
\If{{\sf DrawRandomNumber(0,1)} $\leq \mu$}
\State{$mNode \gets {\sf PickRandomNode} $}
\State{$mNode.{\sf strategy} \gets {\sf DrawRandomNumber(0,1)} $}
\Else
\State $p1\gets{\sf PickRandomNode} $ \Comment{Choose a node u.a.r.}
\State $p2 \gets {\sf ChooseWeightedNeighbor(p1)}$ \Comment{Choose one of p1's neighbors, weighted by payoff}
\State $p1.{\sf strategy} \gets p2.{\sf strategy} $
\Comment{Update strategy of p1 }
\EndIf
\State{time $\gets $ time+1}
\EndWhile
\end{algorithmic}
\end{algorithm}

\begin{algorithm}[t]
\caption{\sf Moran Birth-Death process  }
\label{alg:moranbd}
\begin{algorithmic}[1]
\State \textbf{Global Parameters:} $N$, $k$ \Comment{Number of players, maximum updates in one generation}
\State{time $\gets 0$}
\While{time < $k$}
\If{{\sf DrawRandomNumber(0,1)} $\leq \mu$}
\State{$mNode \gets {\sf PickRandomNode} $}
\State{$mNode.{\sf strategy} \gets {\sf DrawRandomNumber(0,1)} $}
\Else
\State $p1\gets{\sf ChooseWeightedNode} $ \Comment{Choose a node with probability proportional to its payoff}
\State $p2 \gets {\sf PickRandomNeighbor(p1)}$ \Comment{Choose one of p1's neighbors u.a.r.}
\State $p2.{\sf strategy} \gets p1.{\sf strategy} $
\Comment{Update strategy of p2 }
\EndIf
\State{time $\gets $ time+1}
\EndWhile
\end{algorithmic}
\end{algorithm}

\begin{algorithm}[t]
\caption{\sf Pairwise comparison: Imitation process  }
\label{alg:imitation}
\begin{algorithmic}[1]
\State \textbf{Global Parameters:} $N$, $k$, $\beta$ \Comment{Number of players, maximum updates in one generation, selection strength}
\State{time $\gets 0$}
\While{time < $k$}
\If{{\sf DrawRandomNumber(0,1)} $\leq \mu$}
\State{$mNode \gets {\sf PickRandomNode} $}
\State{$mNode.{\sf strategy} \gets {\sf DrawRandomNumber(0,1)} $}
\Else
\State $p1\gets{\sf PickRandomNode} $ \Comment{Choose "learner" node uniformly at random}
\State $p2 \gets {\sf PickRandomNeighbor(p1)}$ \Comment{Choose "role model" node u.a.r.}
\State{$\pi\gets {\sf GetPayoff(p1)}$} 
\State{$\pi'\gets{ \sf GetPayoff(p2)}$}
\State{$\varrho=\frac{1}{1+e^{-\beta(\pi'-\pi)}}$}
\If{{\sf DrawRandomNumber(0,1)} $\leq \varrho$}
\Comment{Update strategy of p1 to p2's with prob. $\varrho$}
\State{$p1.{\sf strategy} \gets p2.{\sf strategy} $}
\Else
\State{$p1.{\sf strategy} \gets p1.{\sf strategy} $} 
\EndIf
\EndIf
\State{time $\gets $ time+1}
\EndWhile
\end{algorithmic}
\end{algorithm}

\begin{compactenum}
    \item[(a)] \textbf{Genetic evolution:} On one hand, we consider genetic evolution as described by the \emph{Moran process}. In this context, we analyze two different variants: a 
    \emph{death-birth (DB)} and \emph{birth-death (BD)} scenario, respectively (cf~\cite{taylor2004evolutionary} and Algs.~\ref{alg:morandb} and~\ref{alg:moranbd}). In the DB scenario, a node is picked to die in each time step of the process. The vacancy is then filled by a copy of one of its neighbors, with the probability of one node being the replacement in some (possibly non-linear) way proportional to its payoff, such that nodes with higher payoffs (or rather, fewer losses) have a higher chance of being chosen as the replacement. In the BD scenario, meanwhile, first a node is picked for reproduction in each round with probability proportional to its payoff. This node subsequently chooses one of its neighbors uniformly at random and replaces it with a copy of itself. After every update, payoffs are recomputed. 
    To visualize an example of such a process, we illustrate the DB scenario in Fig.~\ref{fig:moran}.
    
    \item[(b)] \textbf{Cultural evolution:} 
    On the other hand, we also consider "cultural" evolution through imitation dynamics in the form of the pairwise comparison process (\cite{traulsenPC} and Alg.~\ref{alg:imitation}). Here, a focal player picks a neighboring "role model" node uniformly at random in every time step, observes its payoff, and switches to its strategy with probability 
\vspace*{-0.3cm}
\begin{equation}
\varrho=\frac{1}{1+e^{-\beta(\pi'-\pi)}}
\end{equation}
\vspace*{-0.4cm}

where $\pi'$ is the payoff of the neighbor and $\pi$ the node's own payoff. This function is parameterized by the \emph{selection strength} $\beta \geq 0$, which is a measure for the noise in the update probability, and with it, how much the payoff difference influences the dynamics. Thus, for $\beta=0$, updating is random with probability $\varrho=1/2$, whereas for $\beta > 0$, strategies with lower costs are preferred.
\end{compactenum}

These processes are simulated until \emph{timeout}: once there have been $k$ update steps, we calculate the average welfare of the population (which is the average sum of payoffs), as well as the average count of how often the system visited the different states, and return. 

\begin{figure}[t]
  
  \centering
  \includegraphics[width=0.5\linewidth]{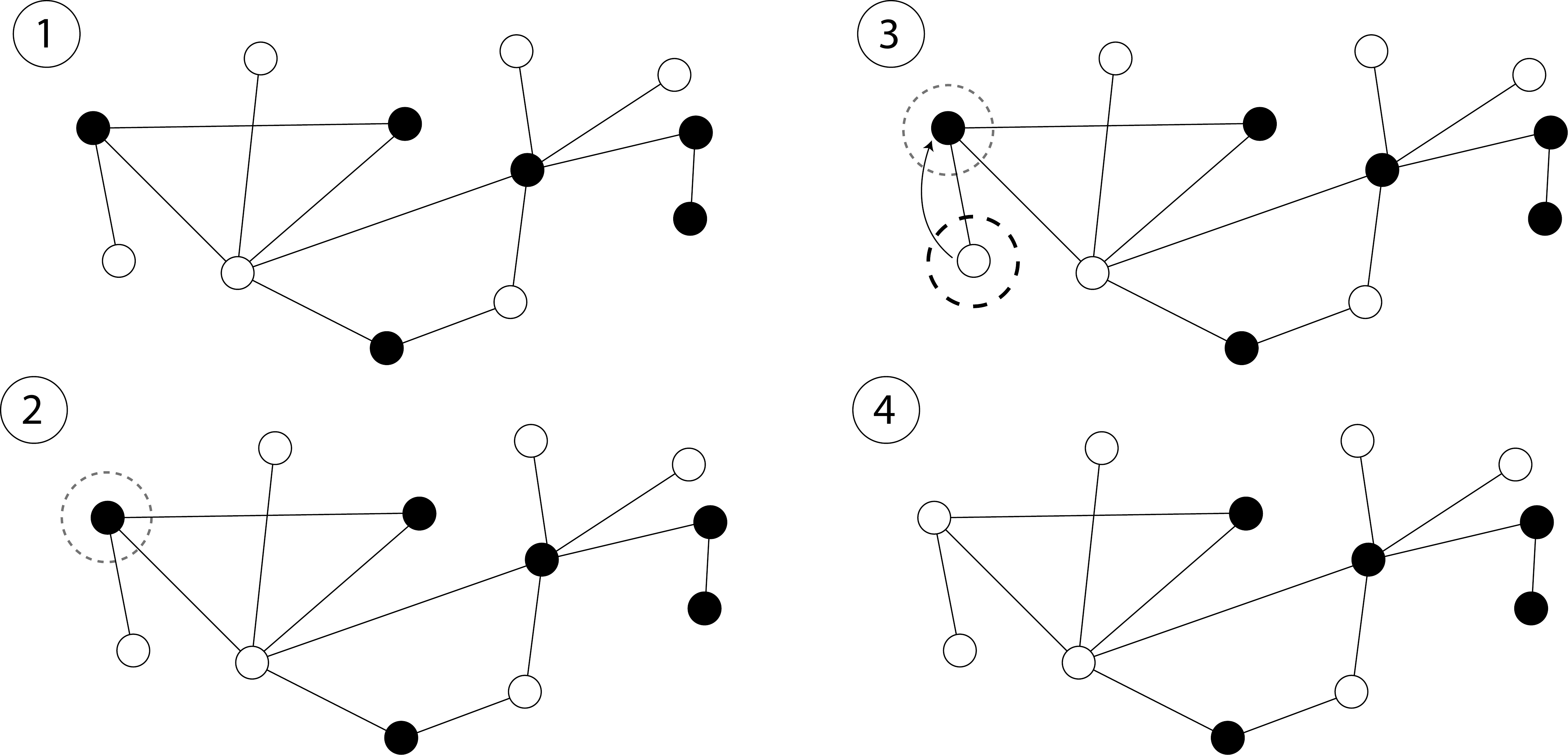}
  \caption{We illustrate the evolutionary dynamics given by the Moran Death-Birth process. Step~1: The nodes in the network shown use either Strategy~1 (white) or Strategy~2 (black). Step~2: One random node -- in our example, a black one -- is picked for death (visualized by the grey dashed line). Step~3: One of the nodes that neighbor the new vacancy is picked for reproduction, depending on its payoff. Here, this is a white node. Step~4: The reproducing node has passed on its strategy, such that there is one more white node on the graph.}
  \label{fig:moran}
\end{figure}

With this dynamic game, we have a model that does not use the assumption of nodes having full information or memory; at most, they need to compare their payoff with the payoffs of their neighbors. 

\subsection{The evolutionary price of anarchy}

In the analysis of an evolutionary game, a most fundamental question 
of interest is to which distribution different dynamics converge to, 
for various topologies and parameters. Such a stationary distribution of the Markov chain underlying the evolutionary dynamics contains the long-run frequencies of states, which correspond to the long-run probabilities of finding the system in the different states. Since we obtain an ergodic process on the space of all possible population compositions (hereon called configurations)  by our elementary updating rules, this stationary distribution of inoculation states exists and is unique. It is also called the \emph{selection-mutation equilibrium} $\textbf{x}$ of a given evolutionary process, and forms the basis of further analysis. We note that it is the nonzero mutation rate $\mu>0$ that provides ergodicity -- otherwise, the Markov chain would have absorbing states where all nodes are inoculated or none are. 

We can subsequently find the average social cost $\hat{S}$ for any of the dynamics we use, either by averaging over the total cost in each round (when the process is simulated), or multiplying $\textbf{x}$ with the vector $\textbf{R}$ containing the cost of all possible configurations, such that 

\begin{equation}
\label{eq:socialaverage}
    \hat{S}=\mathbf{x}\cdot\mathbf{R}=\sum_{i}x_i\,R_i, 
\end{equation}
\noindent where 
    $R_i= \sum_{j=1}^N\hat{\pi}_j^i$ ,
and $\hat{\pi}_j^i$ is the average payoff of player $j$ in configuration $i$.
We measure the efficiency of a process by comparing the average social cost 
$\hat{S}$ with the optimum $\Omega$. At this point, we introduce the concept of the \emph{evolutionary price of anarchy} $ePoA$ as the ratio of the average social cost (or payoffs) of a process against the social optimum.
Similarly to the static PoA, we hence define, for a particular evolutionary dynamics and assuming negative payoffs:

\vspace{-0.3cm}
\begin{equation}
\label{eq:epoa}
    ePoA = \hat{S}/\Omega.
\end{equation}

For positive payoffs, we define $ePoA = \Omega/\hat{S}$ and note that in both
cases, $ePoA\geq 1$, as the static PoA. In general, this quantity gives an indication which processes are most conducive to spending large fractions of time in configurations that minimize overall cost (or maximize positive payoffs).
We also note that in principle, the evolutionary PoA
can be both smaller or larger than the static PoA. 

Note that the concept of an evolutionary price of anarchy is neither bound to a particular game nor a particular evolutionary dynamics: it can serve as a general framework for analyzing a rich variety of dynamic games with arbitrary update dynamics.

\section{Results and Analysis}\label{sec:virusgame} 

In the following, we will consider pure strategies and the analysis of a setting 
with positive mutation rate  $\mu>0$.
We will first show how to exactly calculate the selection-mutation equilibrium $\textbf{x}$ of the evolutionary process for two fundamental graphs, and then use this to show how the ePoA can differ from a static analysis.

\subsection{Analytical results}

For simple graphs and pure strategies, $a_i \in \{0,1\}$, we can 
 calculate the stationary distribution of the underlying Markov chain under the different dynamics.
We consider two instructive cases here
(similar cases have been studied also in~\cite{Aspnes:2005:ISV:1070432.1070440}),
situated at the opposite ends of the spectrum: the \emph{clique} (which results in perfectly
global information) and the \emph{star} (which is fully local). In these cases, we find exact results. 

\subsubsection{Clique} 
\label{sec:clique}
In a clique, the $(N+1)$ states of the Markov chain are $i=0, \ldots, N$, denoting the number of currently inoculated nodes. Here, the threshold for infection of an arbitrary insecure node is 1, as an infected node automatically spreads the virus to all other not inoculated nodes. We use the entries of the cost vector $ct=[I,V,L=0]$ as the realized negative payoffs of infected ($I$), inoculated ($V$) and unaffected ($L=0$) nodes. For the expected payoffs $\hat{\pi}^i_X$ of nodes using strategy $X$, in the state $i$, with $X$ either $C\, (a_i=1)$ or $D\, (a_i=0)$,  we then find the following simple expressions:
\begin{equation}
    \hat{\pi}^i_C=V
\end{equation}
and
\begin{equation}
    \hat{\pi}^i_D=\frac{i}{N}L+\frac{N-i}{N}I=\frac{N-i}{N}I
\end{equation}

Meanwhile, the expressions for the transition probabilities including mutation (that is, with a mutation rate $\mu>0$) are as follows (for a derivation of Moran process transition probabilities, cf.~e.g.\cite{taylor2004evolutionary}):
\begin{equation}
\label{eq:cliquemuteq}
    P_{i,i+1}=\frac{N-i}{N}\mu\frac{1}{2} + (1 - \mu) p_{i,i+1}
\end{equation}
and
\begin{equation}
    P_{i,i-1}=\frac{i}{N}\mu\frac{1}{2} + (1 - \mu) p_{i,i-1} 
\end{equation}

The terms $p_{i,j}$ generally depend on the dynamics being used. One caveat is that the virus inoculation game leads to expected payoffs $\hat{\pi}^i_X \leq 0$. To be able to plug these terms into the equations for the Moran process probabilities, we use the standard assumption of an \emph{exponential fitness function} (see~\cite{traulsen2008analytical}): expected payoffs are mapped to a fitness with the function $F(x)=e^x$, such that the fitness becomes
\begin{equation}
    f^i_X=e^{s\hat{\pi}^i_X}.
\end{equation} We subsequently set the parameter $s=1$, as is common in the literature. This quantity is now always positive, is larger for smaller costs (or equivalently, larger payoffs), and can be used in the standard Moran probabilities 
(cf.~\cite{nowak06} and Eqns.~\ref{eq:completemoran1}-~\ref{eq:completemoran2} in the Appendix).

Meanwhile, for the pairwise comparison -- imitation dynamics, we can still use the payoffs themselves without transforming them, and get Eqns.~\ref{eq:pcclique1} and~\ref{eq:pcclique2} (see Appendix).

From these transition probabilities, we can calculate the stationary distribution of the process: it is the normalized left eigenvector of the transition matrix, which is the tridiagonal matrix $\mathbf{P}$, see Appendix section~\ref{sec:appendcliquemarkov}.

The mutation-selection equilibrium is then the solution to 
\[
\mathbf{x P}=\mathbf{x}.
\]
It gives the long-run frequencies of all possible states. These frequencies can also be obtained by simulating the process for long enough and counting how often each state occurs. 

To be able to compare the evolutionary price of anarchy with the static price of anarchy, we first need to describe the Nash equilibria of the system. For the complete graph, by using Corollary 1.1, the static analysis predicts one equivalence class of Nash equilibria, $\mathcal{N}$, where exactly $N-t=N-VN/I$ nodes are inoculated. We denote the efficiency of these equilibria as $PoA$, the static price of anarchy.

In order to calculate $ePoA$, we first calculate the average social cost $\hat{S}$. We do so either by averaging over the total cost in each round in our simulations, or taking $\hat{S}=\mathbf{x}\mathbf{R}$ (cf.~Eq.~\ref{eq:socialaverage}). For the complete graph, the vector $\mathbf{R}$ containing the total system cost in all possible configurations, with $i=0,...,N$ inoculated nodes, has the components

\begin{equation}
	R_i=i \hat{\pi}^i_C + (N-i) \hat{\pi}^i_D = i V + (N-i) \frac{N-i}{N}I.
\end{equation}

We also know the cost of the optimum; it is attained in the state with $i^*=\frac{N (2 L-V)}{2 L}$, which is the number of inoculated nodes where $R_i^*=\max_{i} R_i$ holds. The optimal cost is then $\Omega=i^* V +  \frac{(N-i^*)^2}{N}I$. With this, we can use Eq.~\ref{eq:epoa} to measure the efficiency of the different dynamics by finding their corresponding evolutionary price of anarchy as $\hat{S}/\Omega$. We present our analysis below.

Using our evolutionary process, our findings can now be summarized in the following lemma:

\begin{lemma}
\label{lemma:completenash}
For a fixed cost vector $ct=[V,I,0]$, large $N\gtrapprox 30$, any reasonable mutation rates $0< \mu < 0.5$, and intermediate to large selection strength $\beta>1$ (for the pairwise comparison process), we always recover the predicted equivalence class of Nash equilibria, $\mathcal{N}$, as the most abundant state in the selection-mutation equilibrium of both types of processes. That is, the process spends the largest fraction of time in the Nash equilibria, where exactly $t=VN/I$ nodes are inoculated.
\end{lemma}

We note that there is also substantial weight on neighboring states of the Nash equilibria (with $t\pm i$, where $i \in \{1,2,3,...\}$) with worse total welfare, due to the stochasticity of the process. However, the average social cost $\hat{S}$ is not substantially different from the cost of the Nash equilibria, since the weight on neighboring states is symmetrically distributed. The numerical results these insights are based on are provided via Figs.~\ref{fig:completedist},~\ref{fig:robustmut} and~\ref{fig:robustn}, as well as upon request in full.

\begin{figure}[h]
  \centering
  \includegraphics[width=0.4\linewidth]{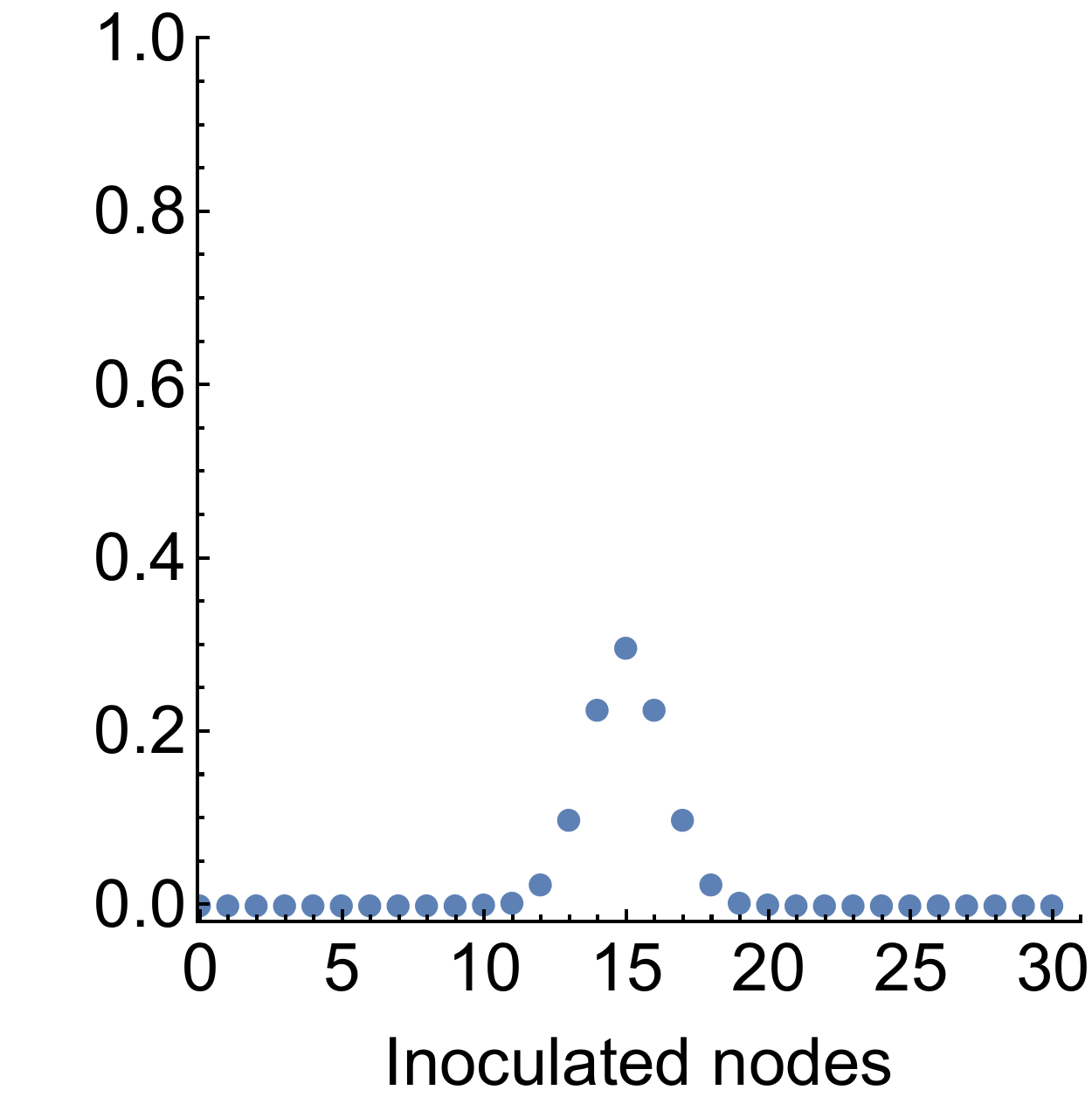}
  \caption{Stationary distribution for the pairwise comparison process on a clique, with $N=30$, $V/I=1/2$ and $\mu=0.001$. We show the number of inoculated nodes (corresponding to the states of the Markov chain) on the x-axis, and their weight in the invariant distribution on the y-axis. The Nash equilibrium, with $N-t=15$ inoculated nodes, is the most abundant state, even though there is significant weight on the neighboring states. The Moran processes lead to qualitatively similar results, but with a stronger dependence on $N$.}  
  \label{fig:completedist}
\end{figure}

What this means for the evolutionary price of anarchy is expressed in the following corollary.

\begin{corollary}[Evolutionary Price of Anarchy for Cliques.]
\label{cor:epoaclique}
The evolutionary price of anarchy $ePoA_{\text{Clique}}$ in a clique with $N$ nodes approaches the efficiency of the Nash equilibrium (the price of anarchy $PoA_{\text{Clique}}$), as $\mu \rightarrow 0$ and $N$ grows large, such that $\lim_{N \to \infty} |ePoA_{\text{Clique}}-PoA_{\text{Clique}}|=0$. 
\end{corollary}

A straightforward argument shows why Lemma~\ref{lemma:completenash} and Corollary~\ref{cor:epoaclique} hold. 
For the base case of the unstructured population on a complete graph, the perfect information setting corresponds to the local information setting, as each node only has distance 1 from every other node. Furthermore, the Markov chain underlying the stochastic evolutionary process is ergodic by $\mu >0$, such that there exists a stationary distribution and the process converges to it. This stationary distribution then places most weight on the Nash equilibrium described in~\cite{Aspnes:2005:ISV:1070432.1070440}, as it is the configuration where nodes have no incentive to switch their strategy, and are aware of this fact, just as in the static case. The stochastic noise becomes smaller as the number of nodes grows larger (which inhibits the system spending too much time in the two extremal states of all nodes inoculated and all nodes insecure), and as the mutation rate becomes smaller (which sharpens the peak of the distribution at the equilibrium). This lets us recover the equilibrium results of~\cite{Aspnes:2005:ISV:1070432.1070440}.

\subsubsection{Star graph} 
\label{sec:star}
For star graphs $K_{1,N-1}$, we can also numerically compute the expected payoffs and the Markov chain properties of the process. The $2N$ states in this case are of the form $(t,l)$, $t \in \{0,1\}$ and $l \in \{0, \ldots, N\!-\!1\}$. The parameter $t$ denotes the inoculation state of the center node, whereas $l$ gives the number of inoculated leaf nodes. 

By using Corollary 1.1 again, we find two equivalence classes of Nash equilibria: class $\mathcal{N}_1$ has $N-t=N-\left \lfloor{VN/I}\right \rfloor$ inoculated leaf nodes (which in our notation is the state $(0,N-t)$), whereas $\mathcal{N}_2$ contains the optimal equilibrium, which features the center as the only inoculated node (state $(1,0)$). We will show that for this highly structured population, the outcome can be quite different from the predictions of the static analysis of the one-shot game. Local evolutionary processes prevent the system from spending too much time in either of the equilibria classes. We detail this in the following paragraphs.

To see this, we first compute the expected payoffs of leaf nodes ($\hat{\pi}^{t,l}_X$, with $X \in \{C,D\}$) and the center node ($\hat{\pi}^{t,l}_{Center}$) in the configurations $(t,l)$:

\begin{equation}
    \hat{\pi}^{t,l}_C=\hat{\pi}^{1,l}_{Center}=V,
\end{equation}
\begin{equation}
    \hat{\pi}^{0,l}_D=\hat{\pi}^{0,l}_{Center}=\frac{N-l}{N}I + \frac{l}{N}L = \frac{N-l}{N}I,
\end{equation}
and
\begin{equation}
    \hat{\pi}^{1,l}_D=\frac{N-1}{N}L + \frac{1}{N}I = \frac{1}{N}I.
\end{equation}

For the Moran and pairwise comparison -- imitation dynamics we derive probabilities $p^{k,m}_{n,o}$ that describe the transitions $(k,n)\rightarrow (m,o)$ without mutation, again with fitness $f^{t,l}_X=e^{\hat{\pi}^{t,l}_X}$. The exact expressions can be found in Section~\ref{sec:appendstartrans}.

We can now again get the transition matrix, with its entries $P^{0,1}_{l,l}$ (Eqns.~\ref{eq:trmatstar1}-~\ref{eq:trmatstar2} in the Appendix), and subsequently the selection-mutation equilibrium $\mathbf{x}$ of the process with its corresponding average system cost. 

Having calculated/simulated the stationary distribution, we observe the following (for results, see Appendix):
\begin{compactitem}
    \item No matter the network size or the process, the Nash equilibria in $\mathcal{N}_1$, $\mathcal{N}_2$ are not abundant states in the stationary distribution, with up to a factor $10^3$ in weight difference to the more frequent states. We instead find a high abundance of costly non-equilibrium states $\mathcal{X}=\{(0,N-t-i)\}$ for some integers $t > i > 0$. There is also substantial weight on the beneficial configurations with low cost $(1,N-t-i)$ for the same values of $i$. 
    \item The equilibrium $\mathcal{N}_1=(0,N-t)$ is typically of far lower frequency than the non-equilibrium states .
    But it still plays more of a role overall than the optimum $\mathcal{O}=\mathcal{N}_2=(1,0)$, which is a rare state at stationarity and almost never visited in the long run.
    
\end{compactitem}

We will now argue why the process exhibits this off-equilibrium behavior. 
First of all, starting from the above observations, it is straightforward to show why the optimum, that is, the state $\mathcal{O}=(1,0)$, cannot be an abundant state in the mutation-selection equilibrium of a star graph.
\begin{lemma}[The optimal Nash equilibrium is rare in the mutation-selection equilibrium.] 
Consider a star graph, with fixed but arbitrary number of nodes $N$. For arbitrary mutation rates $\mu>0$, arbitrary $|V|<|I|$, and any of the three evolutionary processes we consider, the optimal Nash equilibrium $\mathcal{O}=(1,0)$ cannot be an abundant state in the mutation-selection equilibrium.
\end{lemma}

\begin{proof} In fact, the equilibrium is not even reached if not for mutation. 
To see this, consider the states $a=(0,0)$ and $b=(1,1)$, and suppose $\mu=0$. While both these states only differ in one node's strategy from the Nash equilibrium, they cannot serve as starting points for a transition. State $a$ is absorbing in this scenario, as there is no inoculated node to pass on its strategy. Meanwhile, in the state $b$, the one inoculated leaf node cannot change its strategy without the center node being not inoculated -- the terms $p^{1,1}_{l,l-1}$ are always zero. This however leads to the opposite of the Nash equilibrium we are trying to reach. Thus, only a nonzero mutation rate can provide access to this state. At the same time, it is clear that the transitions from $\mathcal{O}$ do not need mutation to reach the neighboring states $a$ and $b$, which leads to a higher probability to leave the state than to reach it. This makes $\mathcal{O}$ unsustainable in the long run.  
\end{proof}

The following lemma states the corresponding result for the other Nash equilibria.
\begin{lemma}[The Nash equilibria from the equivalence class $\mathcal{N}_1$ cannot form an abundant state in the mutation-selection equilibrium.]
\label{offequilibrium}
Consider a star graph, with fixed but arbitrary number of nodes $N$. For arbitrary mutation rates $\mu$, arbitrary $|V|<|I|$, and any of the three evolutionary processes we consider, the Nash equilibria of the form $(0,N-t)$ cannot be an abundant state in the mutation-selection equilibrium. Instead, pairs of states of the form $(0,N-t-i)$ and $(1,N-t-i)$, $i \in \mathcal{I}$ with the set $\mathcal{I} \subseteq \mathbb{N}_0$ depending on $N, V, I, \mu$ and the evolutionary dynamics used, act as sinks of the process. 

\end{lemma}
\begin{proof}
\noindent We use three facts arising from the transition matrix of the process. We argue here for the case of the pairwise comparison process, as it exhibits behavior that is less dependent on network size $N$, but argumentation for the Moran BD scenario is very similar:
\begin{compactenum}
    \item  As the distance $i$ increases, the probabilities $u_{(i,i+1)}$ and $r_{i}$ decrease for $i \geq 1$, as it becomes more unlikely for an insecure node to pass on its strategy to a secure one. However, they do not decrease at the same rate. Rather, their ratio increases exponentially in $i$, such that at some point it is far more likely for the center to become inoculated than for a leaf to switch to insecure. The ratio is
    \begin{equation}
        \frac{r_{i}}{u_{(i,i+1)}}=\frac{e^{\frac{\beta i (-I)}{N}}}{N-1}
    \end{equation}
    It is 1 for the critical value of $i=i^{(1)}_{crit}$. 
    \item As $i$ increases, the probabilities $q_{(i,i-1)}$ and $s_{i}$ decrease at a constant rate, as it becomes overall less likely that an inoculated node will be chosen to pass on its strategy. This ratio is 
    \begin{equation}
        \frac{s_{i}}{q_{(i,i-1)}}=\frac{e^{\frac{\beta (I-N V)}{N}}}{N-1}
    \end{equation}
    We see that $s_i$ is always larger than $q_{(i,i-1)}$, as it is always more likely for one of the leaf nodes to get picked for reproduction if the center is inoculated.
    \item As $i$ increases, the difference $p_{(i,i-1)}-q_{(i,i+1)}$ decreases and eventually changes sign. We find the ratios 
    
    {\footnotesize
    \begin{equation}
        \frac{u_{(i,i+1)}}{q_{(i,i-1)}}=-\frac{\left(e^{\frac{\beta I}{n}-\beta V}+1\right) (N (I-V)-i I)}{\left(e^{-\frac{\beta i I}{N}}+1\right) (-i I-N V+I)}
    \end{equation}} 

    These ratios change from $u/q > 1$ to $u/q < 1$ at a particular value of $i=i^{(2)}_{crit}$, which is the break-even point where moving forward in the lower level of the Markov chain with the center node inoculated is equally probable to moving in the opposite direction in the upper level of the chain.  
\end{compactenum}
Putting together these three facts, we end up with the following explanation for the behavior of the process.

\begin{compactenum}
    \item Start in the Nash equilibrium $\mathcal{N}_1=(0,N\!-\!t)$. 
    \item Setting $i=0$ in Eq. (39), the random walk now has a $(N-1)$ times higher probability of leaving $\mathcal{N}_1$ with a transition to the state $(0,N-t-1)$ than it does of leaving towards $(1,N-t)$.
    \item The probability $u$ of decreasing the number of inoculated leaf nodes in the states $(0,N\!-\!t\!-\!i)$ then gradually decreases with larger $i$, and decreases faster than the probability $r$ of moving towards the state $(1,N\!-\!t\!-\!1)$. This means that after $i^{(1)}_{crit}$, where the two probabilities are equal, a move to a state with the center being inoculated becomes more likely. In our example, $i^{(1)}_{crit}=1$, such that $(0,7)$ is the breakeven point, and the first state where $u$ is significantly different from $r$ is $(0,6)$.
    \item However, the probability of a transition from $(1,N\!-\!t\!-\!1)$ to $(1,N\!-\!t\!-\!i\!+\!1)$ only slowly increases with increasing $i$. This leads to oscillations between opposite states $(1,N\!-\!t\!-\!i)$ and $(0,N\!-\!t\!-\!i)$ until the process escapes -- and the probability of this happening is larger for the transition on the upper level from $(0,N\!-\!t\!-\!i)$ to $(0,N\!-\!t\!-\!i\!-\!1)$, as long as the ratio in Eq. (43) $u/q >1$.  
    \item The further away the process moves from $\mathcal{N}_1$, the longer it takes to escape from the oscillations between the pairs of states. This is most pronounced in the pairs ($(0,N-t-i^{(2)}_{crit})$, $(1,N-t-i^{(2)}_{crit})$) corresponding to the threshold value of $i^{(2)}_{crit}$ where moving forward in $i$ on the upper level and moving backward in the lower level are equal, as well as ($(0,N-t-i^{(2)}_{crit}-1)$, $(1,N-t-i^{(2)}_{crit}-1)$). It is intuitive that the latter pair acts as the strongest sink, since the probabilities of leaving it in either way are smallest. Finding $i^{(2)}_{crit}=5$, the two pairs -- $\{(0,5),(1,5)\}$ and $\{(0,4),(1,4)\}$ -- are circled in Fig.~\ref{fig:starchain}, .  
    \item It is more likely that leaving the last sink happens via the transition $(1,N-t-i_{crit}-1)\rightarrow (1,N-t-i_{crit})$. Then, the cycles repeat, starting from the other side. This makes it intuitive why the process will not find its way back easily to~$\mathcal{N}_1$.
\end{compactenum}   
\end{proof}

The proof is additionally visualized by Fig.~\ref{fig:starchain}, where we show the case of the pairwise comparison process on a star graph with $N=12$ and $V/I=1/3$.

\begin{figure}[t]
  \centering
  \includegraphics[width=0.6\linewidth]{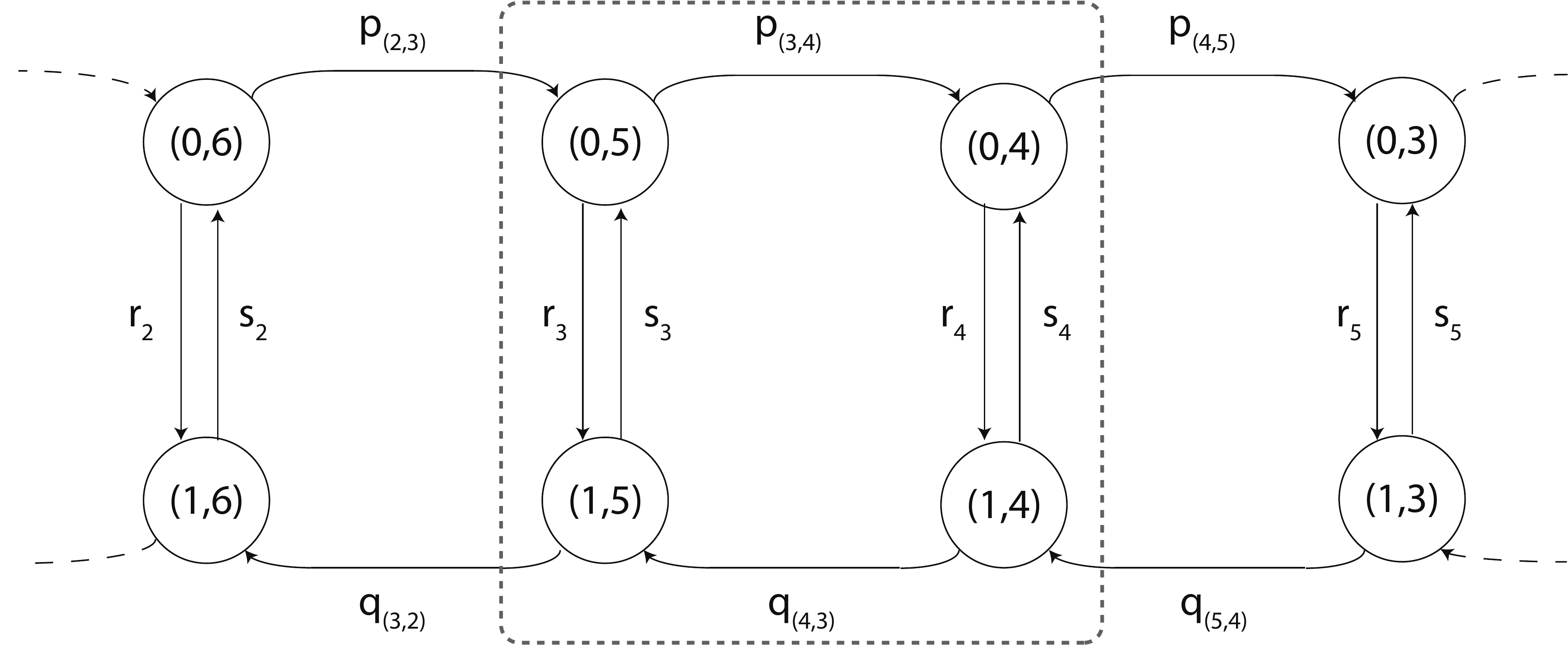}
  \caption{Crucial section of the Markov chain underlying evolutionary processes with local information on the star graph $K_{1,11}$. The ratios $r/u$, $s/q$, and $u/q$ of transition probabilities and their dependency on the distance $i$ from the Nash equilibrium $(0,8)$ has the two-dim. random walk do oscillations between pairs of states $\{(0,5),(1,5)\}$ and $\{(0,4),(1,4)\}$, and thus be trapped in increasingly constricted cycles. }  
  \label{fig:starchain}
\end{figure}

\sloppy


Consider the process as a two-dimensional random walk, defined by our transition probabilities $p^{k,m}_{n,o}$ (Eq.~\ref{eq:star1}-~\ref{eq:star2}) in the limit of the mutation rate $\mu \rightarrow 0$. Let $t=VN/I$; the Nash equilibrium $\mathcal{N}_1$ is then the state $(0,N-t)$, as discussed above. In the example, we have $t=4$, such that $\mathcal{N}_1=(0,8)$. 

\noindent For easier readability, we use the notation $\{u_{(i,i+1)}, q_{(i,i-1)}, r_i, s_i\}$ for \\ $\{p^{0,0}_{N\!-t\!-i\!,N\!-\!t\!-\!i\!-\!1},p^{1,1}_{N\!-\!t\!-\!i,N\!-\!t\!-\!i\!+\!1}, p^{0,1}_{N\!-\!t\!-\!i,N\!-\!t\!-\!i},p^{1,0}_{N\!-\!t\!-\!i,N\!-\!t\!-\!i}\}$  (cf. Eqns~\ref{eq:star1}-~\ref{eq:star2}). That is, $u$ gives the probability of moving one further step away from the Nash equilibrium $\mathcal{N}_1$ by one leaf node switching to insecure; $q$ the probability of moving one step closer to $\mathcal{N}_1$ by one leaf node switching to secure; $r$ the probability of the center switching to secure, and $s$ the probability of the center switching to insecure. The parameter $i$ can be thought to be the distance to $\mathcal{N}_1$. Note that this two-dimensional random walk has a defined direction; there is no possibility to increase $i$ in the lower level of the chain (where the center is inoculated), and no possibility to decrease $i$ in the upper level (where the center is insecure). 

\noindent We show in the Appendix that due to the setup of these transition probabilities, the random walk gets trapped in increasingly constricting cycles as it moves away from $\mathcal{N}_1$
In the example of Fig.~\ref{fig:starchain}, these states form the set $\{(0,5),(0,4),(1,5),(1,4)\}$, with the most weight on $(0,4)$ and $(0,5)$.

We thus have shown that using local information only, the system spends a high fraction of time in states that are not Nash equilibria, and will not reach the optimum. What does this mean for the evolutionary price of anarchy?

Seeing that the abundant states carry a high social cost compared to the optimum and also the worse equilibria in $\mathcal{N}_1$, it is already intuitive that $ePoA$ will be larger than $PoA$ in the star graph, as long as the mutation rate $\mu$ is sufficiently small ($\mu \lessapprox 0.005$). Indeed, using the stationary distribution $\mathbf{x}$ to calculate the social cost $\hat{S}$ as in Eq.~\ref{eq:socialaverage}, and then the evolutionary price $ePoA$ as in Eq.~\ref{eq:epoa}, we find that the evolutionary process has to settle for a relatively high $\hat{S}$, and with it, a high $ePoA$ (see Figs.~\ref{fig:robustmut}-~\ref{fig:robustn}). We summarize this in the following corollary:

\begin{corollary}[Evolutionary Price of Anarchy for Star Graphs.]
\label{cor:epoastar}
For small mutation rates $\mu \lessapprox 0.005$, arbitrary $N$ and arbitrary $|V| < |I|$, the evolutionary price of anarchy $ePoA_{\text{Star}}$ in a star graph $K_{1,N-1}$ with $N$ nodes is at least equal to or higher than the static price of anarchy $PoA_{\text{Star}}$.  That is, $ePoA_{\text{Star}}-PoA_{\text{Star}}\geq 0$. 
\end{corollary}

\smallskip

\begin{figure}[t]
  \centering
  \includegraphics[width=0.4\linewidth]{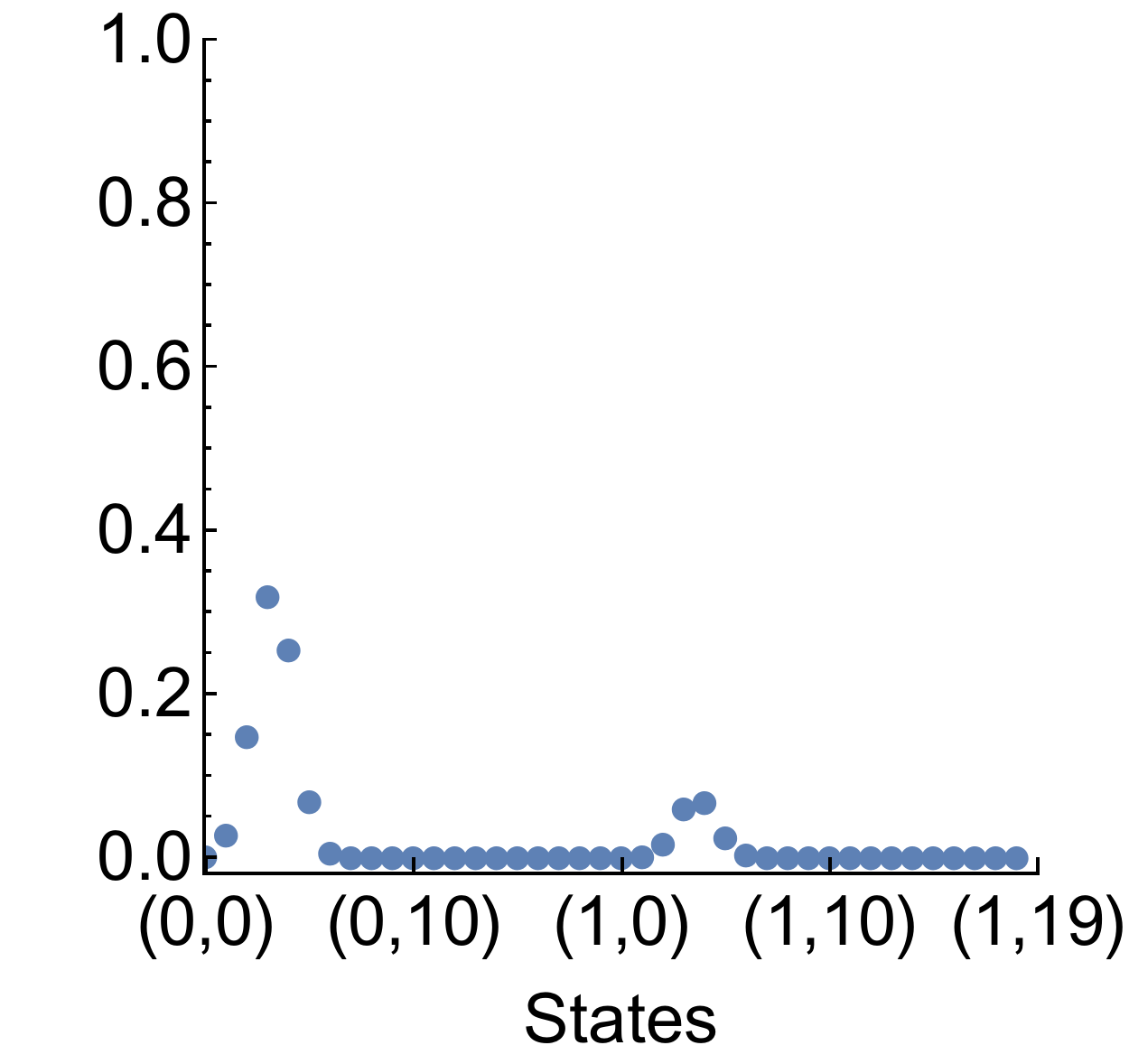}
  \caption{Stationary distribution for the pairwise comparison process on a star, with $N=20$, $V/I=1/2$ and $\mu=0.001$. We show the state $(i,j)$ (where $i \in \{0,1\}$ and $j\in[0,19]$) on the x-axis, and their weight in the invariant distribution on the y-axis. The two Nash equilibria, $(0,10)$ and $(1,0)$ are not abundant states; rather, we find most weight on states with less inoculated nodes, cf. Section~\ref{sec:star}. The pairwise comparison process is again the most efficient dynamics.}  
  \label{fig:stardist}
\end{figure}

\noindent We note that the exact evolutionary price in relation to the static price of anarchy is determined by parameter choices. This means that with a mutation rate $\mu \gtrapprox 0.005$, for some choices of cost vector and network size, it is possible to achieve a slightly lower \emph{average} cost than can be achieved by the worst possible solution to the one-shot game due to the higher mutation rate facilitating contributions from states that lead to a high total payoff (see Fig.~\ref{fig:robustmut}).

However, our results let us conjecture that for a local information model with reasonably small mutation rates, we cannot hope to do much better on average than the worst Nash equilibrium in highly structured networks (like a star), much less reach the optimum, such that paying (at least) the price of anarchy is not only a theoretical possibility, but also a realized fact. 

\section{Simulation of more complex topologies}\label{sec:appendsimresults}

With the algorithms introduced above, we are able to simulate the process also for more complicated graphs. While numerical analysis is usually impossible in these cases - it is intuitive that obtaining numerically precise results on the stochastic process will be harder with increasing graph size and less inherent symmetry in the graph, as the resulting Markov chain can have as many as $2^N$ states already for pure strategies - we can always compute the average social welfare by simulating the different dynamics long enough. We can even feasibly find the stationary distribution for graphs that have some inherent symmetry (which prevents the Markov chain from being of size $2^N$). We will exemplify this by briefly discussing the evolutionary price of anarchy for two such graphs, which are highly symmetric and built from cliques and stars as the two basic components previously analyzed. 

\begin{figure}[t]  
  \centering
  \includegraphics[width=0.5\linewidth]{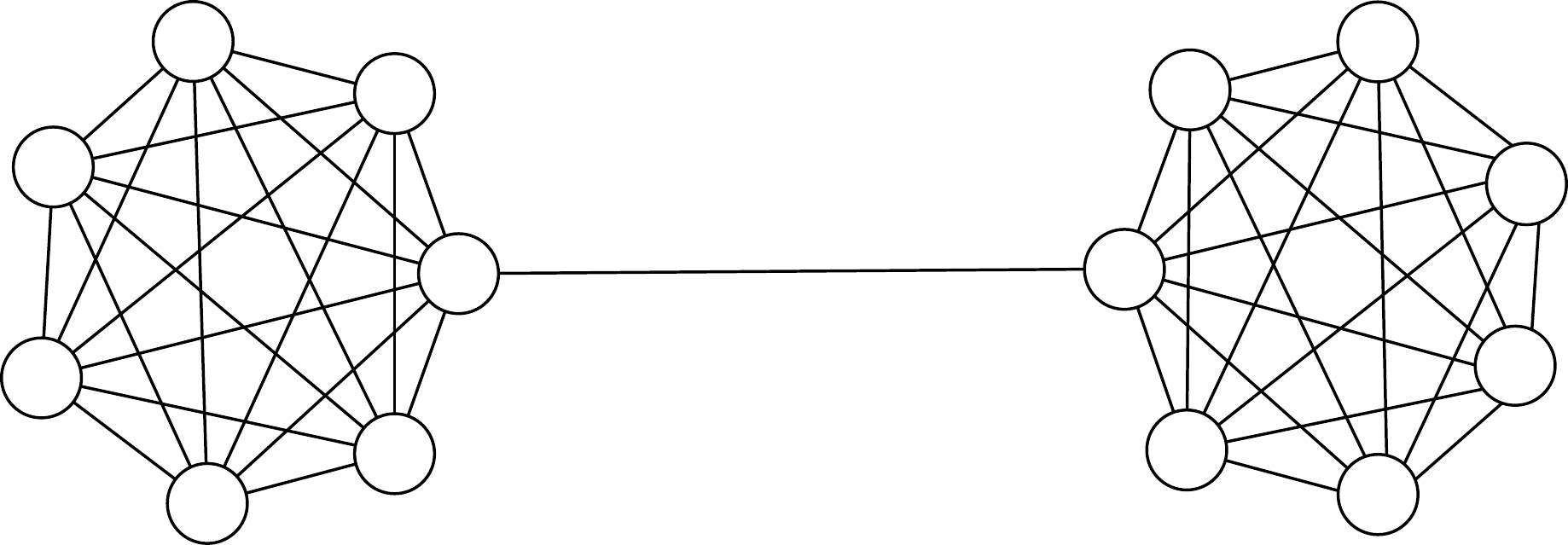}
  \caption{The 2CLIQUE graph, which consists of two complete graphs that are connected at one of their nodes (``hubs'').}
  \label{fig:2clique}
\end{figure}

\begin{figure}[t]
  \centering
  \includegraphics[width=0.5\linewidth]{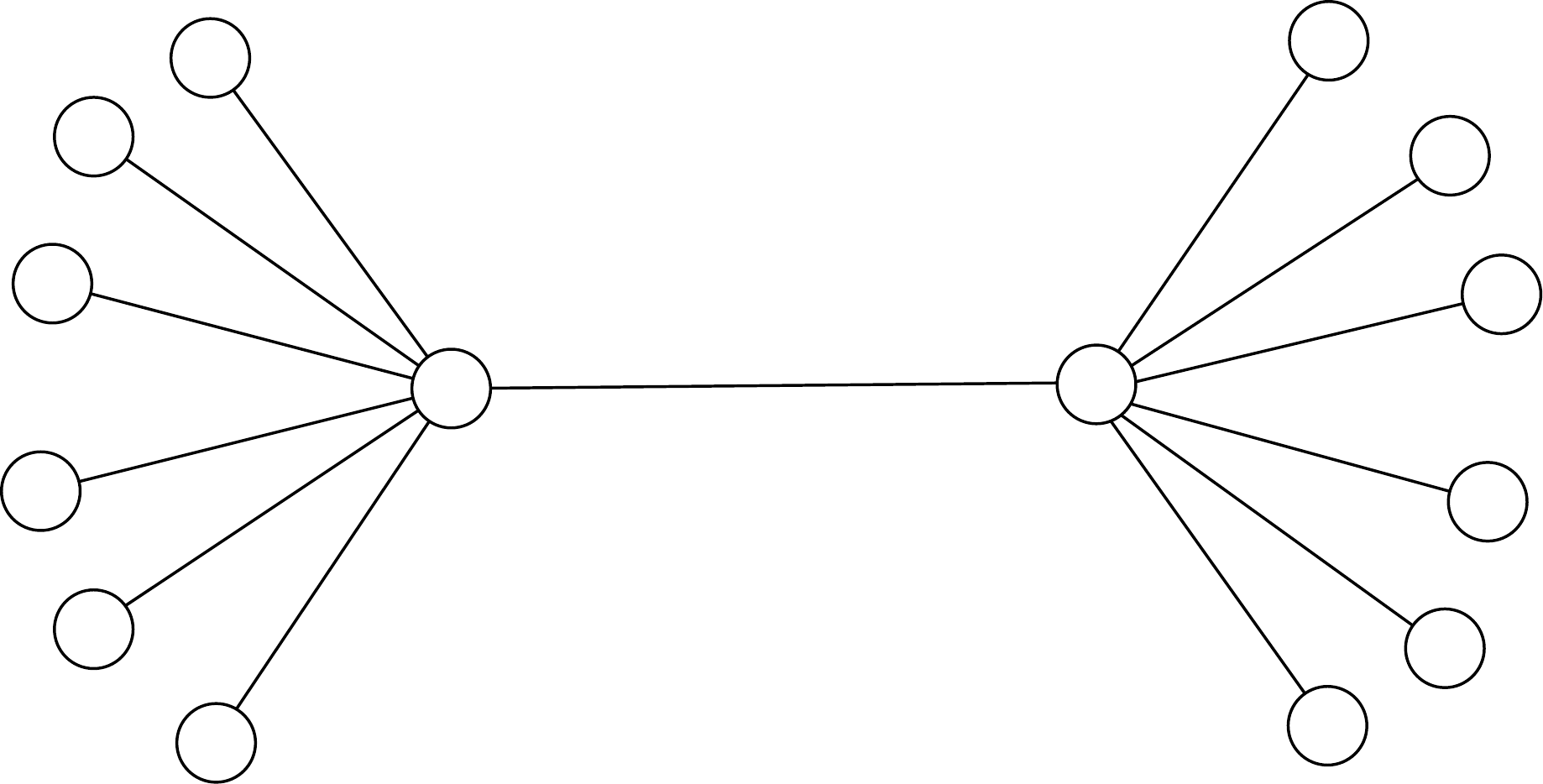}
  \caption{The 2STAR graph, which consists of two star graphs that are connected at their nodes.}
  \label{fig:2star}
\end{figure}

\subsection{2-clique network} 

An example of such a graph is depicted in Fig.\ref{fig:2clique}: two cliques, connected by a single path. The states of the Markov chain are now $(a,b,c,d)$, with $a,c \in \{0,1\}$ denoting the inoculation state of the two nodes where the cliques are joined (subsequently called the hubs), and $b,d \in \{0,...N-1\}$ denoting how many of the $N-1$ remaining nodes on each side are inoculated.  For this graph, seeing that the underlying Markov chain of the process only has $N^2$ states, it is indeed possible to also find both the selection-mutation equilibrium and the average social cost $\hat{S}$ by simulation. To be able to calculate $ePoA$, and compare it to the efficiencies of other equilibria as well as the static price of anarchy $PoA$, we also compute the social welfare in all possible configurations (vector $\vec{R}$), and find the optimum as the maximum over $\vec{R}$. This section describes our findings after running our simulations with $\mu=0.001$ for the network sizes $N=\{10,12,20,50\}$ and the realized cost vectors $ct=\{[-2,-1,0],[-3,-1,0]\}$. 

\begin{enumerate}
    \item First considering static Nash equilibria as the baseline, we again use Corollary 1.1 with $t=V/I$ and find two equivalence classes $\mathcal{N}_1, \mathcal{N}_2$ of Nash equilibria. Class $\mathcal{N}_1$ contains the states $(0,p,0,q)$, with $p+q=N-t$ , which have both hubs insecure, and $p$ respective $q$ inoculated nodes in the remainders of the two cliques. The other class of equilibria, $\mathcal{N}_2$, is composed of the states $(1,N/2-t-1,r,s)$ and $(r,s,1,N/2-t-1)$ with $r+s=N/2-t$, which are the states where at least one hub is always inoculated. In the case where $N$ is divisible by $I$, the two equilibria classes are equivalent, giving the same cost. However, when $N$ is not divisible by $I$, $\mathcal{N}_2$ is more efficient with respect to the overall cost, such that it is the cost of the equilibria in $\mathcal{N}_1$ which is used to calculate the $PoA$. 
    \item  In evolution, we find an $ePoA>1$ for all three network sizes and all three evolutionary processes -- the optimum is not an abundant state. It is also the case that $ePoA \geq PoA$, making the average cost slightly higher than the static price of anarchy, even though $ePoA$ becomes smaller with increasing network size.  In all the tested scenarios, the system does not spend a substantial fraction of time close to social optima and does not make up for partially costly off-equilibrium behavior it exhibits otherwise. We again find that the Moran processes show behavior that is slightly dependent on the network size $N$, whereas the pairwise comparison process gives more consistent results even for smaller $N$, and also exhibits the smallest value of $ePoA$. 
    
\end{enumerate}

We have seen in Corollary~\ref{cor:epoaclique} that for one clique on its own, we recover the Nash equilibria in our evolutionary process. However, the behavior of such a process on the 2CLIQUE graph leads to an outcome that differs both from the optimum and the static predictions for Nash equilibria. Intuitively, this is due to the link between the two hubs acting as a bottleneck for information flow.

\subsection{2-star network} 

Another example of a symmetric graph is depicted in Fig.~\ref{fig:2star}: two stars, joined at their hubs. Here, our findings are the following:

\begin{enumerate}
    \item We again begin by finding the static Nash equilibria and the static $PoA$. In this case, there are three equivalence classes of equilibria: $\mathcal{N}_1$ corresponds to its counterpart in the 2CLIQUE graph -- it contains states $(0,p,0,q)$ (defined analogously to above), with $p+q=N-t$, and is again the class with the most costly equilibria. It therefore is used for the static price of anarchy. Class $\mathcal{N}_2$ consists of the states $(1,0,s,0)$, with $s=0$ if $t \geq N/2$ and $s=1$ otherwise. Lastly, $\mathcal{N}_3$ exists if $t < N/2$, and features states of the type $(1,0,0,N/2-t+1)$. Here, we can also explicitly characterize the optimum: it is always the state $O=(1,0,1,0)$. 
    \item An evolutionary process on the 2STAR graph for the network sizes and $V/I$ ratios tested again leads to an $ePoA \geq PoA$. Seeing that the basic component of the graph -- the simple star -- already exhibits off-equilibrium behavior, this is not too surprising. However, as opposed to the single star, we now observe that the Moran Birth-Death scenario is advantageous for all network sizes, as it leads to the lowest overall $ePoA$. 
\end{enumerate}

\begin{figure}[h]
  \centering
  \hspace*{-1cm}
  \includegraphics[width=1.1\linewidth]{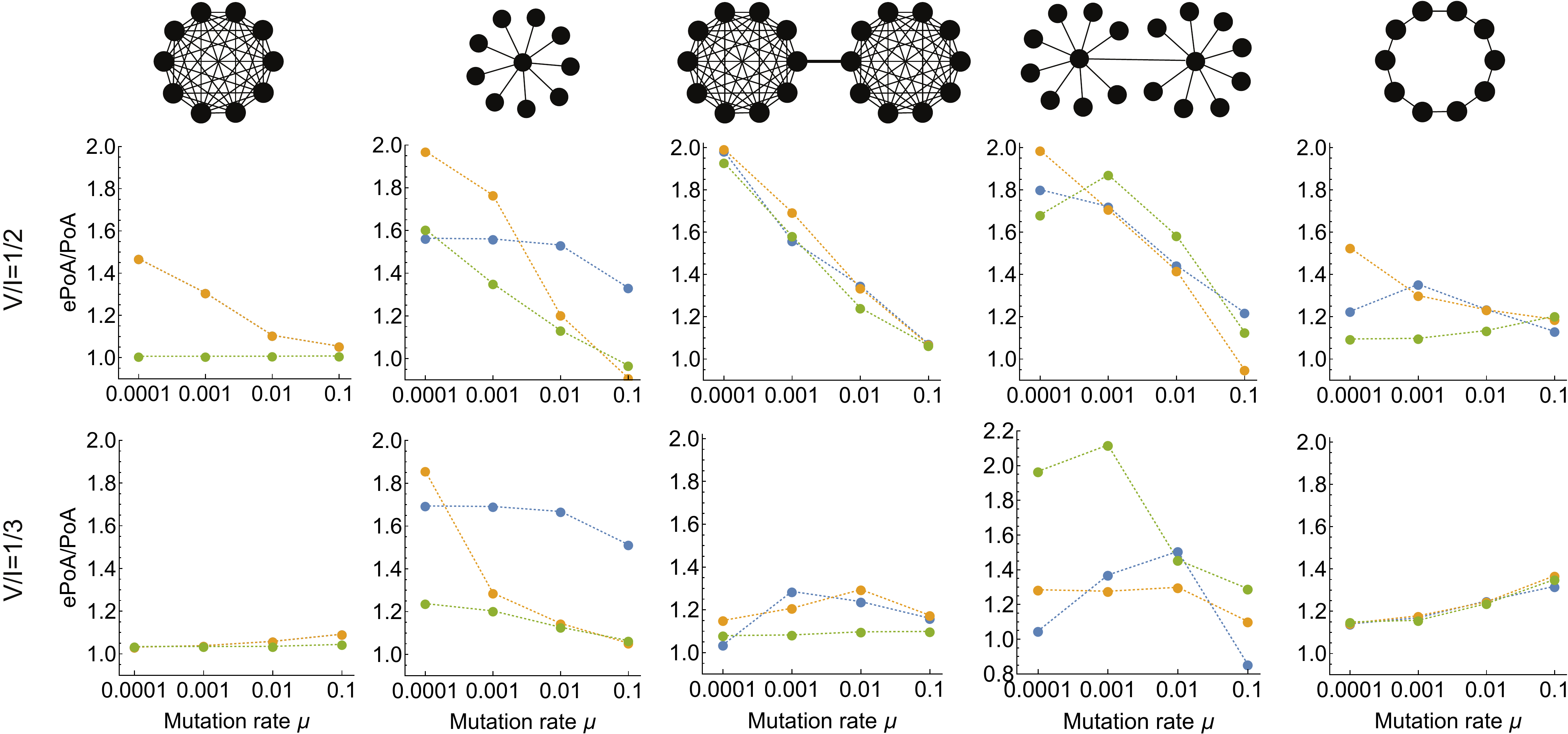}
  \caption{We visualize the ratio of the evolutionary price of anarchy to the price of anarchy, $ePoA/PoA$, for the four discussed topologies, two different values of $V/I$, and varying \emph{mutation rate}. The network size is kept constant at $N=20$. We plot the three different evolutionary dynamics together: Moran-DB (blue), Moran-BD (yellow), and pairwise comparison process (green). We see that different processes show different efficiency, depending on the network topology and the mutation rate, and the behavior of the $ePoA$ does not have to be strictly monotonic - there can be ``sweet spots'' for certain combinations of parameters. Simulations were run for $5*10^5$ iterations. }  
  \label{fig:robustmut}
\end{figure}

\begin{figure}[h]
  \centering
  \hspace*{-1cm}
  \includegraphics[width=1.1\linewidth]{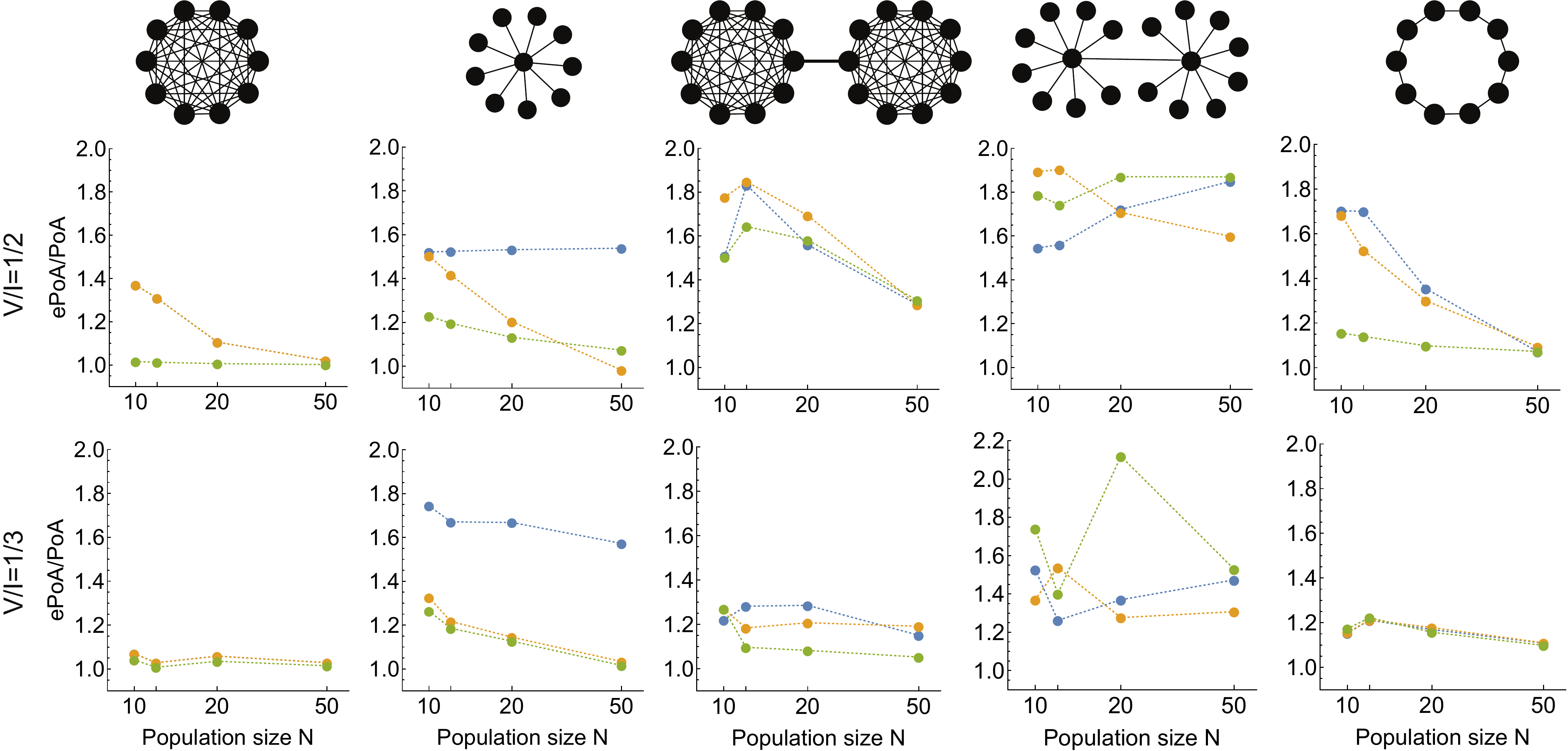}
  \caption{We visualize the ratio of the evolutionary price of anarchy to the price of anarchy, $ePoA/PoA$, for the four discussed topologies, two different values of $V/I$, and varying \emph{network size}. The mutation rate is kept constant at $\mu=0.001$. We again plot the three different evolutionary dynamics together: Moran-DB (blue), Moran-BD (yellow), and pairwise comparison process (green). We see that different processes show different efficiency, depending on the network topology and the network size. Again, the behavior of the $ePoA$ does not have to be strictly monotonic - there can be ``sweet spots'' for certain combinations of parameters. Simulations were run for $5*10^5$ iterations.  }  
  \label{fig:robustn}
\end{figure}

\subsection{Cycle graph}
Another structure of interest is the cycle graph. While this topology is simpler than the previously discussed 2-star and 2-clique networks, the process does not show straightforward behavior due to the large state space with few possible transitions. We start with a description of the states of the Markov chain and the Nash equilibria. 

\paragraph{Encoding of states and size of the state space for network size $N$.}
Due to symmetry considerations, we can encode states corresponding to configurations $\mathbf{a} \in {0,1}^N$ by counting the number of insecure nodes in between every two inoculated nodes, and listing the number of gaps ordered by their size, starting at 1. As an example, we can take a network of size $N=6$ with $i=4$ inoculated nodes in the configuration $010111$, which corresponds to the state $[4,2,0,0,0,0,0]$, where the first number in the list is simply $i=4$. This is thereby also equivalent to the configurations $101011$, $110101$ and $111010$. To enumerate and order these states, we can describe the sets of states with $i$ inoculated nodes more generally as $\{i\}$-integer partitions of $N-i$, where an $\{i\}$-partition of $x$ is the set of partitions of $x$ with length $\leq i$. This way, we get a natural ordering on states - from the smallest partition length to the largest-, and find a simple way to enumerate them. In the previous example, the corresponding partition of $N-i=2$ is $(1,1)$. The total number of states in a network with $N$ nodes is then
\begin{equation}
2+\sum_{i=1}^{N-1}\sum_{k=1}^i p_k(N-i),
\end{equation}  
where $p_k(x)$ is the partition function, giving the number of integer partitions of $x$ with length $k$.
\paragraph{Nash equilibria.} For a state to be a Nash equilibrium, the largest insecure component cannot be larger than $t=VN/I$. For this condition to hold, the number of inoculated nodes has to be $i^*=\left \lceil{\frac{N}{t+1}}\right \rceil $, and the Nash equilibrium states is then the $\{i\}$-partition of $N-i^*$, where the largest part is at most $t$. In our toy example above, using $V/I=1/2$, those are the states $[2,1,0,1,0,0,0]$ and $[2,0,2,0,0,0,0]$. 
\paragraph{Results.} We run simulations for $N=\{10,12,20,50\}$ and $\mu=\{0.0001,0.001,0.01,0.1\}$, and again find the $ePoA$ to be higher than the $PoA$, making the cycle a suboptimal topology for efficient inoculation due to the fact that clusters of inoculated/insecure nodes can not be broken if not for mutation. Intuitively speaking, a Nash equilibrium with e.g. 2 inoculated nodes that are not next to each other will not easily be reached, as moving an inoculated node from point A to B has to involve mutation: starting from a cluster of two inoculated nodes, two imitation events and two mutation events are necessary to move the second node two steps away from the first one. This becomes more unlikely as $\mu \rightarrow 0$. However, a surprising result can be seen in Fig.~\ref{fig:robustmut}, for $V/I=1/3$: the three evolutionary processes behave more or less the same in this case, and the efficiency is monotonic in the mutation rate, just as in the complete graph. An interesting direction for future work would thus be a more thorough analysis of the role of regular graphs in the behavior of stochastic evolutionary processes.

\section{Related work}\label{sec:relwork}

The term price of anarchy was introduced by Koutsoupias and Papadimitriou
in~\cite{koutsoupias1999worst}, however, the idea of measuring inefficiency of equilibrium is
older~\cite{dubey1986inefficiency}. 
The PoA was first mainly studied for network models (see~\cite{nisan2007algorithmic}
for an overview), but literature now covers a wide spectrum, 
from health care to basketball~\cite{roughgarden2015price}. 
The price of anarchy comes in different variations, also considering equilibria
beyond Nash (e.g., the Price of Sinking~\cite{goemans2005sink}).

While in many text book examples, Nash equilibria are typically ``bad'' and highlight
a failure of cooperation (e.g., Prisoner's Dilemma, tragedy of the commons, etc.),
research shows that in many application domains, 
even the worst game-theoretic equilibria are often 
fairly close to the optimal outcome~\cite{roughgarden2015price}. 
However, these examples also have in common that users have 
\emph{full} information about the game. 
While the study of \emph{incomplete} information games also has a long
tradition~\cite{harsanyi1967games}, much less is known today 
about the price of anarchy in games with partial information 
based on local interactions.
In general, it is believed that similar extensions are challenging~\cite{roughgarden2015price}. 
An interesting line of recent work initiated the study of Bayes-Nash equilibria~\cite{singh2004computing}
in games of incomplete information (e.g.,~\cite{singh2004computing}), 
and in particular the 
Bayes-Nash Price of Anarchy~\cite{roughgarden2015price}.

Another interesting concept is 
the stochastic price of anarchy~\cite{chung2008price},
addressing the issue that some PoAs are sensitive to
small perturbations, motivating the study of more stable
notions. However, while the stochastic price of anarchy is also a dynamic concept, it requires
even stronger assumptions on players: beyond current information
(which is derived implicitly) the players also need knowledge of
the execution \emph{history}, to decide on their best action.
Also the dynamic games based on learning dynamics~\cite{cesa2006prediction},
such as regret minimization~\cite{hart2000simple} or
fictitious play~\cite{hofbauer2002global}, 
require players to keep track of historical information. 
All these results hence do not apply to the evolutionary games we study here, as our strategies do not require the nodes to have memory, only local information about their neighbors' payoffs.
To the best of our knowledge, this paper is the first to
consider the application of the price of anarchy to evolutionary
games with this kind of limited information and memoryless players, also providing a novel perspective on the discussion
of the impact of limited information in games. 

In particular, we apply our framework
in a case study, a specific \emph{network game}
where players interact topologically.
Such games have already received much attention in the
past: the price of anarchy was originally introduced with such
games in mind~\cite{roughgarden2015price}.
However, such games are still not well-understood
from an evolutionary perspective. In general, evolutionary games on graphs
are considered challenging to analyze,
even if the underlying topology is simple. See~\cite{allen2014games} for a review on this area.

More specifically, we revisited
Aspnes et al.'s virus inoculation game~\cite{Aspnes:2005:ISV:1070432.1070440}
in this paper.
Traditional 
virus propagation models studied
virus infection 
in terms of birth and death rates of viruses~\cite{epi-model-1,epi-model-2}
as well as through local interactions on networks~\cite{computer-epi-1,computer-epi-2,computer-epi-3} (e.g., 
the Internet~\cite{internet-epi-1,internet-epi-2}).
An interesting study is due to Montanari and
Saberi~\cite{montanari2009convergence} who compare game-theoretic
and epidemic models, investigating differences 
between the viral behavior of the spread of viruses, new technologies,
and new political or social beliefs. 

Aspnes et al.'s model is motivated by percolation theory~\cite{kesten1982percolation}
as infected node infect all unprotected neighbors,
and already led to much follow-up work, e.g., also
considering the impact of a small number of malicious players~\cite{podc-malice},
or considering social network contexts where users account
at least for their friends' costs~\cite{friendship-1,friendship-2}. There also exists literature
on how to encourage highly connected nodes to inoculate~\cite{aspnes2011towards}.
However, we are not aware
In this paper, we extended Aspnes et al.'s model
accounting for local information and 
evolving players (the two main open questions stated in~\cite{Aspnes:2005:ISV:1070432.1070440}),
revisiting and refining their conclusions.

\section{Conclusion}\label{sec:conclusion}

This paper introduced the evolutionary price of anarchy 
to study equilibrium behavior of simple agents interacting
in a distributed system, based on local information. 
We showed that for memoryless agents, 
the resulting equlibria can be significantly different from 
their static equivalent,
and that Nash equilibria are sometimes assumed only very infrequently. 

We believe that our work can provide a novel perspective on the discussion
of the impact of limited information in games. In particular, it opens several interesting avenues for future
research. Our model still comes with several limitations 
due to the well-known notorious difficulty of analyzing evolutionary multi-player games on graphs. 
In particular, it will be interesting to analyze the ePoA
for additional topologies and more general interaction models (beyond memoryless),
as well as to explore randomized (mixed) strategies.
It would also be interesting to prove or disprove our conjecture that
processes based on imitation dynamics always result in a
ePoA which is higher than the PoA in the virus inoculation game.

{\balance

 }
 
\appendix

\section{Clique: Markov chain transition probabilities}
\label{sec:appendcliquemarkov}

The standard Moran probabilities for the complete graph (clique) look the following:
\begin{equation}
\label{eq:completemoran1}
    p_{i,i+1}=\frac{N-i}{N}\,\frac{i\, f^i_C}{i\, f^i_C+(N-i)\, f^i_D},
\end{equation}
and 
\begin{equation}
\label{eq:completemoran2}
    p_{i,i-1}=\frac{i}{N}\,\frac{(N-i)\, f^i_D}{i\, f^i_C+(N-i)\, f^i_D}.
\end{equation}

The pairwise comparison process does not use exponential fitness, but the original payoffs:
\begin{equation}
\label{eq:pcclique1}
    p_{i,i+1}=\frac{N-i}{N}\,\frac{i}{N}\,\frac{1}{1+e^{-\beta( \hat{\pi}^i_C- \hat{\pi}^i_D)}},
\end{equation}
and 
\begin{equation}
\label{eq:pcclique2}
    p_{i,i-1}=\frac{N-i}{N}\,\frac{i}{N}\,\frac{1}{1+e^{-\beta( \hat{\pi}^i_D- \hat{\pi}^i_C)}}.
\end{equation}

The transition matrix $\mathbf{P}$ for the process with mutation looks the following:
\[
\setlength{\arraycolsep}{2pt}
\renewcommand{\arraystretch}{0.8}
P=\begin{pmatrix}
    1-P_{0,1}       & P_{0,1} & 0 & \dots & \dots & 0 \\
    P_{1,0}       & 1-P_{1,0}-P_{1,2} & P_{1,2} & \dots &\dots & 0\\ 
    \vdots & \vdots & \vdots &\ddots &\vdots&\vdots \\
    0       & 0 & 0 & \dots & \dots &P_{N-1,N} \\ 
    0       & 0 & 0 & \dots & \dots & 1-P_{N,N-1}
\end{pmatrix}
\]

The terms $P_{i,j}$ correspond to the terms in Eqns.~\ref{eq:cliquemuteq}f.

In the version of the complete graph that corresponds to the standard case of a well-mixed population (where we have a positive probability of the same node being chosen for reproduction and death), the two Moran processes are equivalent. We note that while Eqns.~\ref{eq:completemoran1}-~\ref{eq:completemoran2} correspond to the standard Moran process on complete graphs as described in the literature~\cite{nowak06}, it is however simple to change the expressions such that there is no possibility of self-replacement for the nodes. This does not change our qualitative results, but does introduce a difference between birth-death and death-birth processes. 
We remark that the difference between $ePoA$ and $PoA$ is always slightly smaller for the pairwise comparison process as compared to the Moran processes. For large values of $N$, both kinds of stochastic processes lead to an outcome with the same average efficiency as the Nash equilibria. Meanwhile, for smaller $N$, the outcomes are somewhat different on first glance: the Moran process has the system drift towards the two extremal configurations, where either all nodes are inoculated or none are. This comes with an overall increased social cost, and is due to the inherent randomness of the stochastic process. This effect can then be mitigated by using the pairwise comparison process with intermediate to strong selection strength $\beta$ instead, or increasing the parameter $s$ in the transformation of payoffs into fitness. This lets us again recover the Nash equilibria as the most abundant states, and with them the corresponding efficiency. 

\section{Star graph transition probabilities}
\label{sec:appendstartrans}
It is intuitive that the two Moran processes on the star graph will not lead to the same dynamics, unlike in the case of the complete graph: for the clique, the global minimal cost and the minimal cost of a neighbor coincide, giving equal probabilities for the best node overall to reproduce and for the best neighboring node to reproduce. This does not hold for the star graph, such that the two Moran processes are no longer equivalent (cf.\cite{hadjichrysanthou2011evolutionary}. The probabilities are the following:

\paragraph{Death-birth process}
\begin{equation}
\label{eq:star1}
    p^{0,1}_{l,l}=\frac{1}{N}\,\frac{l\,f^{0,l}_C}{l\,f^{0,l}_C+(N-l-1)\,f^{0,l}_D}
\end{equation}
\begin{equation}
    p^{1,1}_{l,l+1}=\frac{N-l-1}{N}
\end{equation}
\begin{equation}
    p^{1,0}_{l,l}=\frac{1}{N}\,\frac{(N-l-1)\,f^{1,l}_D}{l\,f^{0,l}_C+(N-l-1)\,f^{0,l}_D}
\end{equation}
\begin{equation}
    p^{0,0}_{l,l-1}=\frac{l}{N}
\end{equation}
\begin{equation}
    p^{0,0}_{l,l+1}=p^{1,1}_{l,l-1}=0
\end{equation}
\begin{equation}
    p^{k,m}_{n,o}= 0 \,\,\, \forall (n,o): |n-o| > 1
\end{equation}

\paragraph{ Birth-death process}
\begin{equation}
    p^{0,1}_{l,l}=\frac{l\,f^{0,l}_C}{f^{0,l}_{Center}+l\,f^{0,l}_C+(N-l-1)\,f^{0,l}_D}
\end{equation}
\begin{equation}
    p^{1,1}_{l,l+1}=\frac{f^{1,l}_{Center}}{f^{1,l}_{Center}+l\,f^{1,l}_C+(N-l-1)\,f^{1,l}_D}\,\frac{N-l-1}{N-1}
\end{equation}
\begin{equation}
    p^{1,0}_{l,l}=\frac{f^{0,l}_{Center}}{f^{0,l}_{Center}+l\,f^{0,l}_C+(N-l-1)\,f^{0,l}_D}\,\frac{l}{N-1}
\end{equation}
\begin{equation}
    p^{0,0}_{l,l-1}=\frac{(N-l-1)\,f^{1,l}_D}{f^{1,l}_{Center}+l\,f^{1,l}_C+(N-l-1)\,f^{1,l}_D}
\end{equation}
\begin{equation}
    p^{0,0}_{l,l+1}=p^{1,1}_{l,l-1}=0
\end{equation}
\begin{equation}
    p^{k,m}_{n,o}= 0 \,\,\, \forall (n,o): |n-o| > 1
\end{equation}

\paragraph{Pairwise comparison -- imitation process}
\begin{equation}
    p^{0,1}_{l,l}=\frac{1}{N}\frac{l}{N-1}\frac{1}{1+e^{-\beta(\hat{\pi}^{0,l}_C-\hat{\pi}^{0,l}_{Center})}}
\end{equation}
\begin{equation}
    p^{1,1}_{l,l+1}=\frac{N-l-1}{N}\frac{1}{1+e^{-\beta(\hat{\pi}^{1,l}_{Center}-\hat{\pi}^{1,l}_D)}}
\end{equation}
\begin{equation}
    p^{1,0}_{l,l}=\frac{N-l-1}{N-1}\frac{1}{N}\frac{1}{1+e^{-\beta(\hat{\pi}^{1,l}_D-\hat{\pi}^{1,l}_{Center})}}
\end{equation}
\begin{equation}
    p^{0,0}_{l,l-1}=\frac{l}{N}\frac{1}{1+e^{-\beta(\hat{\pi}^{0,l}_{Center}-\hat{\pi}^{0,l}_C)}}
\end{equation}
\begin{equation}
    p^{0,0}_{l,l+1}=p^{1,1}_{l,l-1}=0
\end{equation}
\begin{equation}
\label{eq:star2}
    p^{k,m}_{n,o}= 0 \,\,\, \forall (n,o): |n-o| > 1
\end{equation}

The transition matrix has the entries
\begin{equation}
\label{eq:trmatstar1}
    P^{0,1}_{l,l}=\frac{\mu}{2}\,\frac{1}{N} + (1-\mu) p^{0,1}_{l,l}
\end{equation}
\begin{equation}
    P^{1,1}_{l,l+1}=\frac{\mu}{2}\,\frac{N-l-1}{N} + (1-\mu) p^{1,1}_{l,l+1}
\end{equation}
\begin{equation}
    P^{1,0}_{l,l}=\frac{\mu}{2}\,\frac{1}{N} + (1-\mu) p^{1,0}_{l,l}
\end{equation}
\begin{equation}
\label{eq:trmatstar2}
    P^{0,0}_{l,l-1}=\frac{\mu}{2}\,\frac{l}{N} + (1-\mu) p^{0,0}_{l,l-1},
\end{equation}

Comparing the probabilities for the two Moran processes now, we can indeed see that they will not lead to the same dynamics (e.g. by comparing the two expressions for $p^{0,0}_{l,l-1}$). In fact, we make the following observations:
\begin{itemize}
	\item The different evolutionary dynamics indeed show different behavior, which is reflected in the weight distribution on the states as given by their invariant 		distribution, as well as their evolutionary price of anarchy. We find that for any choice of parameters $N$, $\mu$ and $V/I$, the Moran Death-Birth process leads to the 		worst $ePoA$ of the three described processes. It spends a large fraction of time in the two extremal states (for small mutation rates $\mu<0.05$), or has an almost uniform 		invariant distribution over the states.  
    \item For the other two processes, Moran Birth-Death and pairwise comparison, we also observe off-equilibrium behavior. The outcome of the Birth-Death scenario again depends on 		the network size $N$: for small $N<12$, there is again drift towards the two extremal states $(0,0)$ and $(1,N-1)$ as a result of random noise. The smaller the mutation 		rate, the more pronounced this drift is, and the larger the network size needed to balance it out. The pairwise comparison process however does not show the same strong 	dependence on the network size in its behavior.
\end{itemize}

\end{document}